\documentclass{article}

\usepackage[final]{neurips_2022}

\usepackage[utf8]{inputenc} 
\usepackage[T1]{fontenc}    
\usepackage[colorlinks,citecolor=blue,urlcolor=blue,linkcolor=blue,linktocpage=true]{hyperref}       
\usepackage{url}            
\usepackage{booktabs}       
\usepackage{amsfonts}       
\usepackage{nicefrac}       
\usepackage{microtype}      
\usepackage{xcolor}         

\usepackage{amsmath,amsthm,amssymb,bbm}
\usepackage{mathtools}
\usepackage{cases}
\usepackage{dsfont}
\usepackage{microtype}
\usepackage{subfigure}
\usepackage{algorithm,algorithmic}
\usepackage{color}
\usepackage{appendix}

\newtheorem{theorem}{Theorem}[section]
\newtheorem{lemma}[theorem]{Lemma}

\newtheorem{definition}[theorem]{Definition}

\newcommand{\smartparagraph}[1]{\noindent{\bf #1}\ }

\title{Differentially Private Linear Sketches: Efficient Implementations and Applications}

\author{%
  Fuheng Zhao\thanks{Department of Computer Science, UC Santa Barbara.}~~\thanks{The first two authors contributed equally.}\\
  \texttt{fuheng\_zhao@ucsb.edu} \\
  \And
  Dan Qiao\footnotemark[1]~~\footnotemark[2] \\
  \texttt{danqiao@ucsb.edu } \\
  \And
  Rachel Redberg\footnotemark[1] \\
  \texttt{rredberg@ucsb.edu} \\
  \AND
  Divyakant Agrawal\footnotemark[1] \\
  \texttt{agrawal@cs.ucsb.edu} \\
  \And
  Amr El Abbadi\footnotemark[1]\\
  \texttt{amr@cs.ucsb.edu} \\
  \And
  Yu-Xiang Wang\footnotemark[1]\\
  \texttt{yuxiangw@ucsb.edu} \\
}

\begin{document}

\maketitle
\begin{abstract}
Linear sketches have been widely adopted to process fast data streams, and they can be used to accurately answer frequency estimation, approximate top K items, and summarize data distributions. When data are sensitive, it is desirable to provide privacy guarantees for linear sketches to preserve private information while delivering useful results with theoretical bounds. We show that linear sketches can ensure privacy and maintain their unique properties with a small amount of noise added at initialization. From the differentially private linear sketches, we showcase that the state-of-the-art quantile sketch in the turnstile model can also be private and maintain high performance. Experiments further demonstrate that our proposed differentially private sketches are quantitatively and qualitatively similar to noise-free sketches with high utilization on synthetic and real datasets.

\end{abstract}

\section{Introduction}

Data sketches are fundamental tools for data analysis, statistics, and machine learning~\citep{cormode2020small}. Two of the most widely studied problems in data summaries are frequency estimation and quantile approximation. Many real world applications need to estimate the frequency of each item in the database and understand the overall distribution of the database. These applications include stream processing~\citep{das2009thread,bailis2017macrobase}, database management~\citep{misra1982finding,metwally2005efficient, zhao2021spacesaving}, caching~\citep{zakhary2020cot}, system monitoring~\citep{gupta2016network, ivkin2019qpipe, 10.14778/3450980.3450990}, federated learning~\citep{rothchild2020fetchsgd}, among others.

On one hand, the motivation for data sketch algorithms is to efficiently process a large database and extract useful knowledge, since computing the exact information for a large amount of data is both time and memory intensive. For instance, \citet{munro1980selection} proved that to find the true median of a database with $n$ items using $p$ sequential passes requires at least $\Omega(n^{1/p})$ memory. On the other hand, to protect user-level privacy, privacy-preserving algorithms limit the disclosure of private information in the database so that an observer cannot infer much about an individual. Recent works have shown that data sketches can be integrated with privacy-enhancing technologies to provide insightful information and preserve individual privacy at the same time~\citep{cormode2022current}.

Differential privacy~\citep{dwork2006calibrating} is a widely-accepted definition of privacy. Recently, researchers have observed that some data sketches are inherently differentially private~\citep{blocki2012johnson,smith2020flajolet}, while many other data sketches need modifications to the algorithm to be differentially private. In particular, a substantial amount of literature has focused on differentially private data sketches for tasks such as linear algebra~\citep{upadhyay2014differentially, arora2018differentially}, cardinality estimation~\citep{mir2011pan, pagh2020efficient, dickens2022nearly} and quantile approximation~\citep{tzamos2020optimal, gillenwater2021differentially, alabi2022bounded}.

In this paper, we introduce new differentially private algorithms that support both insertions and deletions for frequency, top k, and quantile approximation. While many data sketches assume an insertion-only model~\citep{greenwald2001space, shrivastava2004medians, karnin2016optimal} or a bounded-deletion model~\citep{jayaram2018data,zhao2021spacesaving,10.14778/3450980.3450990}, our algorithms build on top of linear sketches~\citep{charikar2002finding, cormode2005improved} and operate in the turnstile model, which allows an arbitrary number of insertions and deletions into the database. Earlier, researchers attempted to prove CountSketch~\citep{charikar2002finding} itself preserves differential privacy, but the authors acknowledged that there are issues in the proof~\citep{li2019privacy}. Instead of proving that linear sketches, i.e., both Count-Min and CountSketch, are inherently differentially private, we add a small amount of Gaussian noise at their initialization to provide a privacy guarantee, while maintaining linear sketches' original properties, providing high utility for frequency and top K estimations, and keeping update and query algorithms unchanged. \textcolor{black}{We also demonstrate that our analysis provides the tight uniform bound (Section~\ref{sec:main theory result} and Appendix~E).} In addition, we propose the first differentially private quantile sketch in the turnstile model by leveraging the differentially private linear sketch. Our differentially private sketches can be queried an arbitrary number of times without affecting privacy guarantees based on the post-processing immunity. Following prior works~\citep{choi2020differentially, smith2020flajolet}, we assume ideal random hash functions exist, and this assumption can be replaced in practice by
cryptographic hash functions~\citep{dickens2022nearly}.

\section{Preliminaries}

Consider a database $X = \{ i_t \}_{t \in [N]}$ of $N$ items that are drawn from a large $universe$ of size $U$, such as IPv4 address of size $2^{32}$, and for each insert or delete operations, one item can be inserted into or deleted from the database $X$. To support ordered statistic such as quantile, we assume that the $universe$ is some finite totally ordered data universe. 

\begin{definition} Given a database $X$, the frequency of an item $x$ is $f(x) = \sum_{t=1}^{N} \pi(i_t = x)$ where $\pi$ returns 1 if $\text{i}_t$ is $x$, and 0 otherwise.
\end{definition}

\begin{definition} Given a database $X$ of items drawn from an ordered universe, the rank of an item $x$ is $R(x) = \sum_{t=1}^{N} \pi(i_t \leq x)$ where $\pi$ returns 1 if $\text{i}_t$ is less or equal to $x$ and 0 otherwise.
\end{definition}

Given the large size of $N$, calculating the actual statistics, such as frequency and quantile, is often hard, and hence most applications are satisfied with an {\em approximation}. The {\em randomized frequency estimation problem} takes an accuracy parameter $\gamma$ and a failure probability $\beta$ such that, for any item $x$, $|\widehat{f}(x) - f(x)| \leq \gamma \cdot N$ with high probability $1-\beta$, where $\widehat{f}(x)$ is the estimated frequency and $f(x)$ is the true frequency~\citep{cormode2008finding}. In addition, the {\em randomized quantile approximation problem} also takes an accuracy parameter $\gamma$ and a failure probability $\beta$ such that, for any item $x$, $|\widehat{R}(x) - R(x)| \leq \gamma \cdot N$ with high probability $1-\beta$ where $\widehat{R}(x)$ is the estimated rank and $R(x)$ is the actual rank~\citep{karnin2016optimal}.

\subsection{Differential Privacy}

\begin{definition}
Databases $X$ and $X^\prime$ are neighbors ($X \sim X^\prime$), if they differ in at most one element.
\end{definition}

Through this paper, we assume the $update/replace$ definition of differential privacy instead of $add/remove$ definition of differential privacy, in which one item in $X$ is updated or replaced by another item in $X^\prime$~\citep{vadhan2017complexity}.

\begin{definition}[Differential Privacy \citep{dwork2006calibrating}]
A randomized algorithm $M$ satisfies $(\epsilon, \delta)$-differential privacy ($(\epsilon, \delta)$-DP) if for all neighboring databases $X$, $X^\prime$ and for all possible events $E$ in the output range of $M$, we have
\begin{equation*}
    \mathbb{P}(M(X)\in E) \leq e^\epsilon \cdot \mathbb{P}(M(X^\prime)\in E) + \delta.
\end{equation*}
\end{definition}

When $\delta=0$, $\epsilon$-DP is known as pure DP, and when $\delta>0$, ($\epsilon$, $\delta$)-DP is known as approximate DP.

\begin{definition}[Gaussian Mechanism \citep{dwork2006calibrating}]\label{def:gaussian}
Define the $\ell_2$ sensitivity of a function $f:\mathbb{N}^{\mathcal{X}}\mapsto\mathbb{R}^d$ as
\begin{align*}
    \Delta_2(f)=\sup_{\text{neighboring}\,X,X^{\prime}}\|f(X)-f(X^{\prime})\|_2.
\end{align*}
The Gaussian mechanism $\mathcal{M}$ with noise level $\sigma$ is then given by
\begin{align*}
    \mathcal{M}(X) = f(X) + \mathcal{N}(0, \sigma^2 I_d).
\end{align*}
\end{definition}

Specifically, the Gaussian mechanism is known to satisfy a stronger notion of privacy known as  zero-concentrated differential privacy (zCDP, defined below); zCDP lies between pure and approximate DP and can be translated into standard DP notations, as shown in Lemma~\ref{lem:zcdptodp}. Moreover, zCDP satisfies cleaner composition theorems, as shown in Lemma~\ref{lem:adaptive_com}.

\begin{definition} [zCDP \citep{dwork2016concentrated,bun2016concentrated}] A randomized mechanism $M$ satisfies $\rho$-Zero-Concentrated Differential Privacy ($\rho$-zCDP), if for all neighboring databases $X,X^{\prime}$ and all $\alpha\in (1,\infty)$,
$$D_{\alpha}(M(X)\| M(X^{\prime}))\leq \rho \alpha,$$
where $D_{\alpha}$ is the Renyi divergence~\citep{van2014renyi}.
\end{definition}

\begin{lemma}[Adaptive composition and Post Processing of zCDP \citep{bun2016concentrated}]\label{lem:adaptive_com}
Let $M : \mathcal{X}^n\rightarrow \mathcal{Y}$ and $M^\prime: \mathcal{X}^n\times\mathcal{Y}\rightarrow\mathcal{Z}$. Suppose $M$ satisfies $\rho$-zCDP and $M^\prime$ satisfies $\rho^\prime$-zCDP (as a function of its first argument). Define $M^{\prime\prime} : \mathcal{X}^n\rightarrow \mathcal{Z}$ by $M^{\prime\prime}(x) = M^\prime
(x, M(x))$. Then $M^{\prime\prime}$ satisfies $(\rho+\rho^\prime)$-zCDP.
\end{lemma}

\begin{lemma}[Privacy Guarantee of Gaussian mechanism \citep{dwork2014algorithmic,bun2016concentrated}]\label{lem:gau_mechanism}
Let $f:\mathbb{N}^{\mathcal{X}}\mapsto\mathbb{R}^d$ be an arbitrary d-dimensional function with $\ell_2$ sensitivity $\Delta_2=\sup_{\text{neighboring}\,X,X^{\prime}}\|f(X)-f(X^{\prime})\|_2$. Then for any $\rho>0$, Gaussian Mechanism with parameter $\sigma^2=\frac{\Delta^2_2}{2\rho}$ satisfies $\rho$-zCDP.
\end{lemma}

\begin{lemma}[Converting zCDP to DP \citep{bun2016concentrated}]\label{lem:zcdptodp}
If M satisfies $\rho$-zCDP then M satisfies $(\rho+2\sqrt{\rho\log(1/\delta)},\delta)$-DP.
\end{lemma}

As we use exclusively Gaussian mechanisms and their composition in our proposed algorithms, our method actually satisfies $(\epsilon,\delta)$-DP guarantees with stronger parameters than what is implied by zCDP via techniques from  \citep{balle2018improving,dong2019gaussian}, which reduces the $\epsilon$ parameter by a sizable fraction in typical parameter regimes. We stick to zCDP for clarity and generality, because all our results would apply without changes if we modify the noise into other mechanisms satisfying zCDP, e.g., the Discrete Gaussian Mechanism \citep{canonne2020discrete}.



\subsection{Revisiting Linear Sketches}

			
			
			

\citet{charikar2002finding} proposed the \textbf{CountSketch}, a randomized algorithm that summarizes a database and solves the frequency estimation problem. The CountSketch uses a $d \times w$ array of counters, i.e, C[d, w], where all the counters are initialized to \textbf{zero}, and has two sets of independent hash functions $h$ and $g$. For each row $r$, the hash function $h_{r}$ maps input items uniformly onto $\{1,\dots,w\}$ and the hash function $g_{r}$ maps input items uniformly onto $\{-1, +1\}$. 
For item $x$ with value $v \in \{-1, +1\}$, CountSketch updates $d$ counters, one per each row, based on the hash values such that for a particular row $r$, $g_{r}(x)$ will be added or subtracted to the counter at the $h_{r}(x)^{th}$ index depending on whether $x$ is being inserted or deleted respectively, as shown in Algorithm~\ref{alg:update}. Hence, the update time is $O(d)$. To estimate the frequency of item $x$, CountSketch will output the $\mathrm{median}_{1\leq r \leq d}$ $g_{r}(x)\cdot C[r, h_r(x)]$, as shown in Algorithm~\ref{alg:query}. By updating each row's counter based on the hashed value of either $1$ or $-1$ and reporting the median for query, CountSketch provides an unbiased estimate. To reduce the failure probability of bad estimations, $d$ is set to $O(\log(1/\beta))$ and it uses $O(\frac{1}{\gamma}\log(\frac{1}{\beta}))$ space to solve the frequency estimation problem.

\begin{algorithm}[H]
	\caption{Linear Sketch Update$(x, v)$}
	\label{alg:update}
	\small{
		\begin{algorithmic}[1]
			\STATE {\bfseries Input:} Item $x$ with value $v \in \{-1, +1\}$, counter arrays $C$, and two sets of hash functions $\{h_1, \ldots, h_{C.rows}\}$ and $\{g_1, \ldots, g_{C.rows}\}$.
			
			\FOR{$r \gets 1,2,\ldots,C.rows$}
			    \STATE $C[r,h_r(x)] \gets C[r,h_r(x)] + v \cdot g_r(x)$
			\ENDFOR
			\STATE {\bfseries Output: $C$. } 
			
		\end{algorithmic}
	}
\end{algorithm}

\citet{cormode2005improved} proposed the \textbf{Count-Min} sketch that shares the same initialization, update, and data structure as CountSketch. Count-Min sketch also uses $O(\frac{1}{\gamma}\log(\frac{1}{\beta}))$ space to solve the frequency estimation problem. A major difference is that Count-Min sketch makes all hash functions in set $g$ return positive 1. As a result, to estimate the frequency of item $x$, Count-Min sketch returns $\min_{1\leq r \leq d}C[r, h_r(x)]$ instead of the median, as shown in Algorithm~\ref{alg:query}. In addition, it has the nice property of never underestimating item's frequency. Since linear sketches can approximate an item's frequency accurately, they also solves the top K approximation problem by returning the K items associated with the highest estimated frequency.

\begin{algorithm}[H]
	\caption{Linear Sketch Query$(x)$}
	\label{alg:query}
	\small{
		\begin{algorithmic}[1]
			\STATE {\bfseries Input:} Item $x$, counter arrays $C$, and two sets of hash functions $\{h_1, \ldots, h_{C.rows}\}$ and $\{g_1, \ldots, g_{C.rows}\}$.
		
		\STATE $\mathrm{arr} \gets [\,]$
		
		\FOR{$r \gets 1,2,\ldots,C.rows$}
			\STATE $\mathrm{arr.append}(g_r(x) \cdot C[r, h_r(x)])$
			
	    \ENDFOR
		\STATE {\bfseries Output: $\min(\mathrm{arr})$ for Count-Min or $\mathrm{median}(\mathrm{arr})$ for CountSketch. } 
		\end{algorithmic}
	}
\end{algorithm}

\citet{gilbert2002summarize} made the connection between frequency and quantiles, in which the quantile range can be decomposed into at most $\log U$ dyadic intervals ~\citep{cormode2019answering} and the sum of the estimated frequencies for these intervals gives the estimated rank.~\citet{wang2013quantiles} leveraged the unbiased property of CountSketch and proposed the Dyadic CountSketch (DCS) to estimate the frequencies of each dyadic interval. For more specific details, Appendix~B and~\citep{cormode2020small} provide a comprehensive analysis of quantile sketches.

\section{Private Linear Sketches}

In this section, we present new algorithms for differentially private linear sketches. We highlight that our Private Count-Min and CountSketch only require a different initialization while they share the same update (Algorithm~\ref{alg:update}) and query (Algorithm \ref{alg:query}) with the original Count-Min and CountSketch. Therefore, the implementation of our algorithms is efficient. Below we show our private initialization.

\begin{algorithm}[H]
	\caption{DP Linear Sketch Initialization with Gaussian Noise}
	\label{alg:DP-initialization}
	\small{
		\begin{algorithmic}[1]
			\STATE {\bfseries Input:} Desired accuracy parameter $\gamma$, failure probability $\beta$, and budget for zCDP $\rho$.
			
			\STATE {\bfseries Initialize Counter Arrays}
			
			\STATE $\sigma \gets \sqrt{\log(2/\beta)/\rho}$
			\STATE $E \gets \sqrt{\frac{2\log\frac{2}{\beta}}{\rho}}\cdot\sqrt{\log\frac{\frac{4}{\gamma}\log(\frac{2}{\beta})}{\beta}}$
			\FOR{$r \gets 1,2,\ldots,\log(2/\beta)$}
			    \FOR{$c \gets 1,2,\ldots,1/\gamma$}
			\STATE $C[r,c] \gets \mathcal{N}(0, \sigma^2)$ if Private CountSketch
			\STATE $C[r,c] \gets E+\mathcal{N}(0, \sigma^2)$ if Private Count-Min
			
			    \ENDFOR
			\ENDFOR
			
			\STATE {\bfseries Output: $C$. } 
			
		\end{algorithmic}
	}
\end{algorithm}

In Algorithm~\ref{alg:DP-initialization}, the set of arrays we use is $C$ which consists of $\log(2/\beta)$ arrays with length $1/\gamma$, which has the same space complexity as original Count-Min and CountSketch. Recall that two neighboring databases $X$ and $X^{\prime}$ differ by at most one item. Therefore, after updating all the items respectively, for each corresponding array in $C(X)$ and $C(X^{\prime})$, they differ by at most two elements and the difference is at most $1$. Then the $\ell_2$-sensitivity of the set of arrays $C$ is bounded by 
\begin{equation}
\Delta_{2}=\sqrt{2\log(2/\beta)}.
\end{equation}

By applying the Gaussian Mechanism (Definition \ref{def:gaussian}), we can add \emph{independent} Gaussian noises $\mathcal{N}(0,\sigma^{2})$ to each counter in $C$, where $\sigma=\sqrt{\frac{\log(2/\beta)}{\rho}}$. Due to the privacy guarantee of Gaussian Mechanism (Lemma \ref{lem:gau_mechanism}), it satisfies $\frac{\Delta_{2}^{2}}{2\sigma^{2}}=\rho$-zCDP. 

Define $E(\beta,\gamma,\rho)=\sqrt{\frac{2\log\frac{2}{\beta}}{\rho}}\cdot\sqrt{\log\frac{\frac{4}{\gamma}\log(\frac{2}{\beta})}{\beta}}$, for simplicity, we will use $E$ in Algorithm~\ref{alg:DP-initialization} and the proof in Appendix~A. The private version of Count-Min can be derived by adding \emph{independent} Gaussian noises $\mathcal{N}(E,\sigma^{2})$ to each counter of $C$, while the private version of CountSketch can be derived by adding \emph{independent} Gaussian noises $\mathcal{N}(0,\sigma^{2})$ to each counter of $C$. The private versions of Count-Min and CountSketch are derived by combining Algorithm~\ref{alg:DP-initialization}, Algorithm~\ref{alg:update}, and Algorithm~\ref{alg:query}.


\subsection{Main results about Private Count-Min and CountSketch}
\label{sec:main theory result}

We present the privacy guarantee and utility analysis of our Private Count-Min and CountSketch below. Recall that for each item $x$, we perform update as in Algorithm \ref{alg:update} and query as in Algorithm~\ref{alg:query}. In addition, $\widehat{f}(x)$ is the output estimated frequency and $f(x)$ is the actual frequency. \textcolor{black}{To provide a bound for the additional error due to DP, we define $\widetilde{f}(x)$ to be the non-private estimated frequency (the output of the original Count-Min and CountSketch with the same set of hash functions).} We begin with the properties of Private Count-Min. Note that all the proofs are deferred to Appendix~A.


\begin{theorem}\label{thm:min}
Private Count-Min satisfies $\rho$-zCDP regardless of the number of queries. Furthermore, with probability $1-\beta$, the output $\widehat{f}(x)$ satisfies that 
$$\forall\, x, 0\leq\widehat{f}(x)-\widetilde{f}(x)\leq 2E.$$
In addition, for each item $x$, with probability $1-\beta$, $$0\leq \widehat{f}(x)-f(x)\leq \gamma\cdot N+2E.$$
\end{theorem}



\textbf{Comparison to Count-Min.} Comparing our Theorem \ref{thm:min} with the conclusion in \citep{cormode2005improved}, our Private Count-Min preserves the nice property that the output will not underestimate the frequency with high probability. Furthermore, within the most popular regime where the privacy budget $\rho$ is a constant, the additional error bound due to differential privacy is independent of the size of database $N$, therefore it will become negligible as $N$ goes large.

\textbf{Justification of our Gaussian noise.} Note that with high probability, all the noises we add ($E+\sigma_{i,j}$, $\sigma_{i,j}\sim \mathcal{N}(0,\sigma^2)$) will be non-negative. Therefore, the noise we add and the original error induced by Count-Min will directly sum up and lead to larger error in evaluation. However, we claim that the additional $E$ ensures that with high probability, for all item $x$, the output will not underestimate the actual frequency. This nice property enables the good performance of our Private Count-Min when used in approximate top $k$ task, as shown in Section \ref{sec:evaluation}. 

Next, Theorem \ref{thm:median} shows the properties of Private CountSketch.
\begin{theorem}\label{thm:median}
Private CountSketch satisfies $\rho$-zCDP regardless of the number of queries. Furthermore, the frequency query from Private CountSketch is unbiased and with probability $1-\beta$, 
$$\forall\, x, |\widehat{f}(x)-\widetilde{f}(x)|\leq E.$$
In addition, for each item $x$, with probability $1-\beta$,
$$|\widehat{f}(x)-f(x)|\leq \gamma\cdot N+E.$$
\end{theorem}

\textbf{Comparison to CountSketch.} Comparing our Theorem \ref{thm:median} with the conclusion in \citep{charikar2002finding}, our Private CountSketch preserves the nice property that the output will be an unbiased estimate of the frequency. This property enables our use of Private CountSketch in quantile estimation below (Section \ref{sec:quantile}). Furthermore, within the most popular regime where the privacy budget $\rho$ is a constant, the additional error bound due to differential privacy is independent of the size of database $N$, thus it will become negligible as $N$ goes large.

\subsection{The Uniform Bound of Additional Error}
Theorem~\ref{thm:min} and Theorem~\ref{thm:median} show a uniform bound $\sup_{x}|\widehat{f}(x)-\widetilde{f}(x)| \leq O(E)$ for linear sketches, which upper bounds the additional error imposed on the estimated frequency due to Differential Privacy guarantees. To derive the point-wise bound for $|\widehat{f}(x)-f(x)|$, we combine our result with the point-wise bound for $|\widetilde{f}(x)-f(x)|$~\citep{cormode2005improved, charikar2002finding} (note it is straightforward to apply other analyses on the point-wise bound~\citep{minton2014improved,larsen2021countsketches} due to the triangle inequality). Moreover, in Appendix~E, we demonstrate that our analysis provides the tight uniform bound when items are drawn from a large universe.


\section{Private Quantile Sketches}\label{sec:quantile}
In this section, we apply our Private CountSketch to state of the art quantile sketches in the turnstile model. Our private Dyadic CountSketch can estimate all the quantiles accurately at the same time while ensuring differential privacy.

\subsection{Revisiting DCS}
In \citep{wang2013quantiles}, it is shown that DCS can return all $\gamma$-approximate quantiles with constant probability using space $O\left(\frac{1}{\gamma}\log^{1.5}U\log^{1.5}(\frac{\log U}{\gamma})\right)$. More specifically, the sketch structure here consists of $\log U$ CountSketches, each CountSketch uses a counter arrays $C$, which is $d\times w$ counters. The choice of $d,w$ follows $d=\Theta\left(\log(\frac{\log U}{\gamma})\right)$ and $w=O\left(\sqrt{\log U\log(\frac{\log U}{\gamma})}/\gamma\right)$.

\subsection{Private DCS}
In this work, we aim to estimate the quantiles accurately while preserving privacy. We do this by replacing CountSketch with PrivateCountSketch, which bases on the same structure as CountSketch discussed above. Given the privacy budget $\rho$, the privacy budget of each Private CountSketch is thus $\rho_0=\frac{\rho}{\log U}$, due to composition of zCDP (Lemma \ref{lem:adaptive_com}). The $\ell_2$-sensitivity of each Private CountSketch is $$\Delta_2=O(\sqrt{2d})=O\left(\sqrt{\log\left(\frac{\log U}{\gamma}\right)}\right).$$
To keep the whole algorithm $\rho$-zCDP, it suffices to keep each CountSketch $\rho_0$-zCDP (Lemma \ref{lem:adaptive_com}). Therefore, Gaussian Mechanism (Definition \ref{def:gaussian}) with $\sigma^2=O\left(\log U\log\left(\frac{\log U}{\gamma}\right)/\rho\right)$ ensures $\rho$-zCDP (Lemma \ref{lem:gau_mechanism}). Similar to Lemma \ref{lem:utility}, define $E(\gamma,U)=O\left(\sqrt{\frac{\log U\log\left(\frac{\log U}{\gamma}\right)}{\rho}}\cdot\sqrt{\log\frac{\log U\log\left(\frac{\log U}{\gamma}\right)}{\gamma}}\right)$, we can prove that with constant probability, all the Gaussian noises we add to all $\log U$ CountSketches are bounded by $E$ (for simplicity, we use $E$ to represent $E(\gamma,U)$).

Conditioned on the high probability event above, we prove that for a fixed quantile, the estimated quantile will be accurate with high probability. As has been proven in Theorem \ref{thm:median}, the output estimated frequency is unbiased for any item. Therefore, similar to \citep{wang2013quantiles}, for any item $x$ (corresponding to a fixed CountSketch), we have the output $\widehat{f}(x)$ of that CountSketch satisfies
$$\mathbb{P}\left[\left|\widehat{f}(x)-f(x)\right|>\frac{1}{w}\cdot N+E\right]<\exp\left(-O(d)\right)=O\left(\frac{\gamma}{\log U}\right).$$
By a union bound, with probability $1-\log U\times O\left(\frac{\gamma}{\log U}\right)=1-O(\gamma)$, for any item corresponding to this fixed quantile, the error of CountSketch is bounded by $\frac{1}{w}\cdot N+E$. Conditioned upon this event, by Hoeffding's inequality, with probability $1-O\left(\frac{\gamma}{\log U}\right)$, the sum of $\log U$ such independent errors is bounded by 
\begin{equation}\label{eqn:quantile}
\sqrt{\log U\log\left(\frac{\log U}{\gamma}\right)}\cdot\left(\frac{N}{w}+E\right)=\gamma\cdot N+E^\prime,
\end{equation}
where $E^\prime=O\left(\frac{\log U\log\left(\frac{\log U}{\gamma}\right)}{\sqrt{\rho}}\cdot\sqrt{\log\frac{\log U\log\left(\frac{\log U}{\gamma}\right)}{\gamma}}\right).$ To sum up, for a fixed quantile, with probability $1-O(\gamma)$, the estimating error is bounded by $\gamma\cdot N+E^\prime$.

Finally, apply another union bound on the $\frac{1}{\gamma}$ different quantiles, with constant probability, all the quantiles are estimated accurately (within the error bound \eqref{eqn:quantile}). Note that similar to \citep{wang2013quantiles}, the failure probability here is a constant. For any failure probability $\beta$, we can further increase $d$ by a factor of $\log \frac{1}{\beta}$ to reduce this failure probability to $\beta$. 

\textbf{Take-away of Private DCS.} First, our Private DCS has a same space complexity as the original DCS. In addition, according to \eqref{eqn:quantile}, the additional error bound is proportional to $\frac{\log U\log\frac{1}{\gamma}}{\sqrt{\rho}}$ (ignoring $\log\log$ terms), and independent to the size of database $N$. In the most popular regime where the privacy budget $\rho$ is a constant, the additional error bound only appears as lower order terms, which will become negligible as $N$ goes large.

\section{Evaluation}\label{sec:evaluation}

We have implemented DP linear sketches and DP DCS, and conducted extensive experiments to evaluated the privacy-utility trade-off of our proposed private sketches. The implementations are written in Python with the advantage of fast prototyping and good readability. The code for the following experiments can be found on Github~\footnote{https://github.com/ZhaoFuheng/Differentially-Private-Linear-Sketches}.

\subsection{Data Sets}

The experimental evaluation is conducted using both synthetic and real world data sets. We consider the synthetic Zipf dataset~\cite{zipf2016human} with universe size of $2^{16}$ and the source IP addresses from CAIDA Anonymized Internet Trace 2015 dataset~\cite{passive_2015_pcap} with universe size of $2^{32}$. For each independent run in the experiments, we use an input database size $N = 10^{5}$.
    


\subsection{Metrics}

In all experiments, we average the various metrics over 5 independent runs to minimize the measurement variance. The metrics used in the experiments are: 

\textbf{Average Relative Error}:  Let the set $\Psi$ denotes all unique items in the database. Average Relative Error (ARE) is computed based on $\Psi$ in which $\frac{1}{\Psi} \sum_{e \in \Psi} \frac{ |f(e) - \widehat{f}(e)|}{f(e)}$.

\textbf{F1 Score}: F1 score is the harmonic mean of the precision and recall ($2 \cdot \frac{precision \cdot recall}{precision + recall}$). 

\textbf{Average Rank Error}: For each evenly spaced quantile and its associated item, we average the distance between the true rank and estimated rank.

 We use ARE to evaluate sketch performance on frequency estimation and F1 score to evaluate the sketch's performance in identifying the top 10 items. For quantile approximation, we consider the $m$ evenly spaced quantiles and items. For instance, if $m=1$, we consider the rank error for the median item; if $m=2$, we consider then average rank error for the $33^{rd}$ and $67^{th}$ percentile items. Lower ARE and average rank error, and higher F1 score indicate better approximation.

\subsection{Private Linear Sketches Experiments}

To evaluate the utility of DP linear sketches, we compare the average relative error (ARE) and F1 score for frequency estimation and identify the top 10 items, respectively. As shown in Figure~\ref{fig:FrequencyMain}, the x-axis represents the space budget for each sketch (from 9.2 KB to 147.3 KB), and the y-axis denotes ARE or F1 score. The DP linear sketches use $\rho \in \{0.1, 1, 10\}$ in which lower $\rho$ value indicates more noise need to be added, and all sketches assume $\beta = 1\%$.

For frequency estimation, the performance of our private CountSketch with various privacy budgets is basically equivalent to the performance of the non-private CountSketch. Under different space and privacy budgets, they have minimal difference in ARE for both Zipf and CAIDA datasets, meaning that, while providing strong privacy guarantee, the estimated frequencies are still very accurate. The accurate estimation of private CountSketch is primarily due to the unbiased nature of CountSketch in which, by adding Gaussian noise, the private CountSketch still provides unbiased estimation for an item's frequency as proved in Theorem~\ref{thm:median}. As shown in both Figure~\ref{fig:FEZipf} and Figure~\ref{fig:FECaida}, the performance of private Count-Min degrades when the space budget increases or the privacy budget decreases. This behavior is expected as the upper bound on the frequency error in Theorem~\ref{thm:min} has a dependency on both $\gamma$ and $\rho$. In order to preserve the property of not underestimating an item's frequency, the private Count-Min sketch needs to add larger noise to each counter when the number of counters increases. As a result, the estimated frequencies for low-frequency items become inflated and this in turn decreases the overall accuracy.
\begin{figure}[H]
	\centering     
	\subfigure[Frequency Estimation with Zipf]{\label{fig:FEZipf}\includegraphics[width=0.4\linewidth]{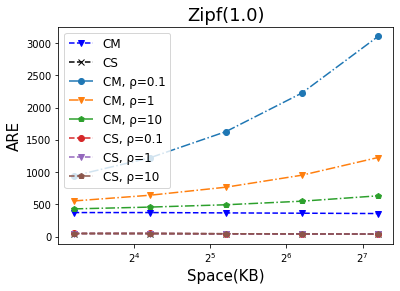}}
	\subfigure[Frequency Estimation with CAIDA]{\label{fig:FECaida}\includegraphics[width=0.4\linewidth]{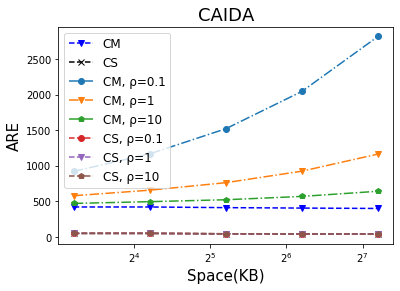}}\\
	\subfigure[Top 10 with Zipf]{\label{fig:TopZipf}\includegraphics[width=0.4\linewidth]{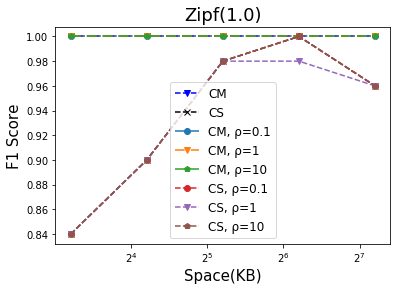}}
	\subfigure[Top 10 with CAIDA dataset]{\label{fig:TopCaida}\includegraphics[width=0.4\linewidth]{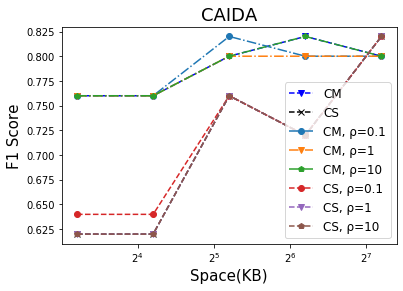}}
	\caption{Comparison of non-private linear sketches and DP linear sketches with various privacy budget under synthetic and real world datasets. The experiments assume $\beta = 1\%$ and $N=10^5$.}
	\label{fig:FrequencyMain}
\end{figure}

For approximate top 10 items, private CountSketch has similar performance to CountSketch. Since both non-private and private CountSketch are unbiased, they may underestimate the frequency of true top K items and decrease the recall.
On the other hand, the property of no underestimation is desirable for approximate top K items. In particular, non-private and private Count-Min sketch score high F1 scores for all datasets. While providing privacy guarantees, private Count-Min achieve 1.0 F1 scores for all space and all privacy budgets in Zipf dataset, as shown in Figure~\ref{fig:TopZipf}.

\subsection{Private Quantile Sketch Experiments}

\begin{figure}[H]
	\centering     
	\subfigure[Quantile Approximation with Zipf]{\label{fig:QuantileZipf}\includegraphics[width=0.4\linewidth]{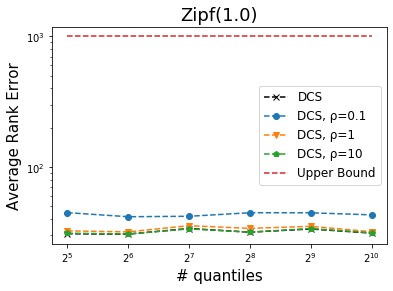}}
	\subfigure[Quantile Approximation with CAIDA]{\label{fig:QuantileCaida}\includegraphics[width=0.4\linewidth]{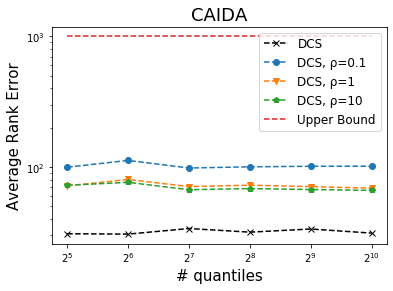}}\\
	\caption{Compare DCS and DP DCS with various privacy budget under synthetic and real world datasets. The experiments assume $\gamma = 1\%$, $N=10^5$, and the desired error upper bound is $10^3$ ($\gamma N$).}
	\label{fig:QuantileMain}
\end{figure}

To evaluate the utility of DP DCS, we compare the average rank error. As shown in Figure~\ref{fig:QuantileMain}, the x-axis represents the number of evenly spaced quantiles, and the y-axis denotes the average rank error. The DP DCS use privacy budget $\rho \in \{0.1, 1, 10\}$  and all sketches assume $\gamma = 1\%$. 

For the quantile approximation, we observe that the increase in the number of evenly spaced quantiles does not impact the average rank error, as shown in both Figure~\ref{fig:QuantileZipf} and Figure~\ref{fig:QuantileCaida}. Since the CAIDA dataset universe size ($2^{32}$) is larger than Zipf dataset universe size ($2^{16}$), the average rank error in the CAIDA dataset is larger than the average rank error in the Zipf dataset. As shown in Equation~\eqref{eqn:quantile}, the error bound has a term depending on the universe size in which a large universe size leads to more error. When the privacy budget decreases, the average rank error increases as more noise needs to be added. Comparing DP DCS with strong privacy ($\rho=0.1$) to DCS, the increase in rank error is relatively small compared to the database size of $10^{5}$. In addition, the desired rank error upper bound is $\gamma \cdot N = 10^{3}$ and all the rank errors are one order of magnitude lower.

\section{Related Works}
\vspace{-0.1in}
\textcolor{black}{In this section, we discuss and compare our results to previous literature on Private Count-Min Sketch~\citep{mir2011pan, melis2015efficient, ghazi2019power}, and the concurrent work on Private CountSketch~\citep{pagh2022improved}. In fact, ~\citet{pagh2022improved} and us both independently discovered the same algorithm for Private CountSketch with differences in the theoretical analysis. To the best of our knowledge, we are the first to present a DP quantile sketch in the turnstile model.}

\smartparagraph{Private Count-Min.}\textcolor{black}{~\citet{mir2011pan} proposed to add Laplace noise into the Count-Min Sketch estimator to derive the number of heavy hitters with Pan-Privacy~\citep{dwork2010pan}. Similarly,~\citet{melis2015efficient} add independent Laplace noise to each counter of the sketch instead of the estimator. However, adding Laplace noise breaks the nice property of never underestimation in Count-Min. In contrast, our private Count-Min guarantees no underestimation with high probability.~\citep{ghazi2019power} added one-sided binomial noise into each counter of the sketch to preserve the property of no underestimation. However, using the Binomial mechanism inherently implies approximate differential privacy~\citep{canonne2020discrete}. In contrast, by using the Gaussian mechanism, our Private Count-Min provides the stronger concentrated differential privacy guarantee.}

\smartparagraph{Private CountSketch.}\textcolor{black}{~\citet{pagh2022improved} and us both independently discovered the same algorithm for private CountSketch. There is a major difference in the analysis and we believe both analyses are valuable, in which~\citet{pagh2022improved} focused on deriving a tight point-wise bound for $|\widehat{f}(x)-f(x)|$, while we focused on deriving a uniform bound for $\sup_{x}|\widehat{f}(x)-\widetilde{f}(x)|$. Our uniform bound for $\sup_{x}|\widehat{f}(x)-\widetilde{f}(x)|$ can derive the point-wise bound for $|\widehat{f}(x)-f(x)|$, by combining our result with any point-wise bound for $|\widetilde{f}(x)-f(x)|$ due to the triangle inequality.~\citep{pagh2022improved} obtains a tighter point-wise bound by using concentration of median instead of triangle inequality. However, \citet{pagh2022improved}'s analysis can not imply the point-wise bound for $|\widehat{f}(x)-\widetilde{f}(x)|$. More detailed comparisons are included in Appendix~C.}

\section{Conclusion}

In this work, we demonstrate that linear sketches can be made differentially private and provide useful information while maintaining their original properties by adding a small amount of Gaussian noise at initialization. In addition, leveraging the private CountSketch, we propose the DP DCS for quantile approximation in the turnstile model. DP DCS achieves low rank errors even for a large data universe. Moreover, for all the proposed algorithms, when the privacy budget is constant, the additional error due to privacy is independent of the database size and the error will become negligible when the database grows larger. \textcolor{black}{Moreover, private linear sketches bring new opportunities for other statistical questions such as the private euclidean distance estimation~\citep{stausholm2021improved} which can be calculate as the dot product of two private linear sketches}. As a result, we believe our proposed algorithms are efficient and practical for real-world systems and enable these systems to perform data analysis and machine learning tasks privately.

\acksection
This work is partially supported by gifts from Snowflake Inc, and NSF grants CNS-1703560, CNS-1815733 and CNS-2048091. The authors thank Rasmus Pagh for a helpful discussion regarding their concurrent work \citep{pagh2022improved}. The authors also thank Adam Smith for clarifying the mergeability in the inherently private  Flajolet-Martin Sketch \citep{smith2020flajolet}.

\bibliographystyle{plainnat}
\bibliography{reference}

\appendix
\clearpage
\section{Missing proofs}\label{sec:proof}
In this section, we present the missing proofs. Recall that  for each item $x$, we perform query as in Algorithm~\ref{alg:query}. In the proof below, we use $\mathrm{arr}^\prime$ to denote the $\mathrm{arr}$ in Algorithm \ref{alg:query} (with private initialization) while $\mathrm{arr}$ denotes the $\mathrm{arr}$ in Algorithm~\ref{alg:query} under non-private initialization (all zero initialization) and the same set of hash functions. In addition, $\widehat{f}(x)$ is the output estimated frequency, $f(x)$ is the actual frequency and $\widetilde{f}(x)$ is the non-private estimated frequency (the output of the original Count-Min and CountSketch with the same set of hash functions). We first present the following Lemma \ref{lem:utility}, which gives a high probability bound for the Gaussian noises we add.

\begin{lemma}[Utility analysis]\label{lem:utility}
If there are $\frac{1}{\gamma}\log(\frac{2}{\beta})$ independent Gaussian noises sampled from $\mathcal{N}(0,\sigma^{2})$ (where $\sigma=\sqrt{\frac{\log(2/\beta)}{\rho}}$), denoted as $\sigma_{ij}$ where $i\in[\log(2/\beta)]$, $j\in[1/\gamma]$, then with probability $1-\frac{\beta}{2}$, for any $i\in[\log(2/\beta)]$, $j\in[1/\gamma]$,
\begin{equation}
|\sigma_{ij}|\leq\sqrt{2}\sigma\cdot \sqrt{\log\frac{\frac{4}{\gamma}\log(\frac{2}{\beta})}{\beta}}=\sqrt{\frac{2\log\frac{2}{\beta}}{\rho}}\cdot\sqrt{\log\frac{\frac{4}{\gamma}\log(\frac{2}{\beta})}{\beta}}.
\end{equation}
\end{lemma}

\begin{proof}[Proof of Lemma \ref{lem:utility}]
The lemma directly results from the concentration inequality of Gaussian distribution and a union bound.
\end{proof}

\begin{theorem}[Restate Theorem \ref{thm:min}]\label{thm:restate_min}
Private Count-Min satisfies $\rho$-zCDP regardless of the number of queries. Furthermore, with probability $1-\beta$, the output $\widehat{f}(x)$ satisfies that 
$$\forall\, x, 0\leq\widehat{f}(x)-\widetilde{f}(x)\leq 2E.$$
In addition, for each item $x$, with probability $1-\beta$, $$0\leq \widehat{f}(x)-f(x)\leq \gamma\cdot N+2E.$$
\end{theorem}

\begin{proof}[Proof of Theorem~\ref{thm:restate_min}]
Differential privacy directly results from the DP guarantee of Gaussian Mechanism (Lemma \ref{lem:gau_mechanism}) and post processing (Lemma~\ref{lem:adaptive_com}).\\

According to Lemma \ref{lem:utility}, we have with probability $1-\frac{\beta}{2}$, for all noises $E+\sigma_{ij}$, where $\sigma_{ij}\sim \mathcal{N}(0,\sigma^{2})$, $i\in[\log(2/\beta)]$, $j\in[1/\gamma]$, it holds that $$0\leq E+\sigma_{ij}\leq 2E.$$
Under the above case that will happen with probability $1-\frac{\beta}{2}$, for all $x$ and the corresponding $\mathrm{arr}$, $\mathrm{arr}^\prime$, it holds that
$$\min(\mathrm{arr})\leq \min(\mathrm{arr}^{\prime})\leq \min(\mathrm{arr}+2E\cdot\mathds{1})\leq \min(\mathrm{arr})+2E.$$
Therefore, for all $x$,
$$\widetilde{f}(x)\leq\widehat{f}(x)\leq \widetilde{f}(x)+2E.$$

For the last conclusion (point-wise bound), by the property of Count-Min~\citep{cormode2005improved}, we have for any item $x$, with probability $1-\frac{\beta}{2}$,
$$0\leq \min(\mathrm{arr})-f(x)\leq \gamma\cdot N.$$
Therefore, for any item $x$, with probability $1-\beta$, we have
$$\widehat{f}(x)=\min(\mathrm{arr}^{\prime})\geq \min(\mathrm{arr})\geq f(x)$$ and $$\widehat{f}(x)=\min(\mathrm{arr}^{\prime})\leq \min(\mathrm{arr})+2E\leq f(x)+\gamma\cdot N+2E.$$
Then the proof is completed by plugging in the definition of $E$.
\end{proof}

\begin{theorem}[Restate Theorem \ref{thm:median}]\label{thm:restate_median}
Private CountSketch satisfies $\rho$-zCDP regardless of the number of queries. Furthermore, the frequency query from Private CountSketch is unbiased and with probability $1-\beta$, 
$$\forall\, x, |\widehat{f}(x)-\widetilde{f}(x)|\leq E.$$
In addition, for each item $x$, with probability $1-\beta$,
$$|\widehat{f}(x)-f(x)|\leq \gamma\cdot N+E.$$
\end{theorem}

\begin{proof}[Proof of Theorem \ref{thm:restate_median}]
First of all, differential privacy directly results from the DP guarantee of Gaussian Mechanism (Lemma \ref{lem:gau_mechanism}) and post processing (Lemma~\ref{lem:adaptive_com}).

Next we claim that the conclusion that $\mathbb{E}\widehat{f}(x)=f(x)$ arises from symmetry. If we replace $\{h_i\}_{i\in[\log(2/\beta)]}$, $\{g_i\}_{i\in[\log(2/\beta)]}$, $\{\sigma_{ij}\}_{i\in[\log(2/\beta)],j\in[1/\gamma]}$ with $\{h_i\}_{i\in[\log(2/\beta)]}$, $\{g_i^\prime\}_{i\in[\log(2/\beta)]}$, $\{-\sigma_{ij}\}_{i\in[\log(2/\beta)],j\in[1/\gamma]}$ where
$g_i^\prime(x)=g_i(x)$ and $g_i^{\prime}(x^\prime)=-g_i(x^\prime),\,\forall\,x^\prime\neq x$, then the outputs under these two cases will be symmetric around $f(x)$ and the probability distribution function at these two cases are identical. Therefore, we have 
\begin{equation}
   \mathbb{E}\widehat{f}(x)=f(x).
\end{equation}

According to Lemma \ref{lem:utility}, we have with probability $1-\frac{\beta}{2}$, for all noises $\sigma_{ij}\sim \mathcal{N}(0,\sigma^{2})$, $i\in[\log(2/\beta)]$, $j\in[1/\gamma]$, it holds that $$|\sigma_{ij}|\leq E.$$
Under the above case that will happen with probability $1-\frac{\beta}{2}$, we will prove that for all $x$, $$|\widehat{f}(x)-\widetilde{f}(x)|\leq E.$$ 
Without loss of generality, we can assume that $\log\frac{2}{\beta}=2k+1$, where $k$ is a positive integer (we can choose $k$ to be the minimum integer such that $2k+1\geq \log\frac{2}{\beta}$).\\
Suppose that $\mathrm{median}(\mathrm{arr}^{\prime})>\mathrm{mediam}(\mathrm{arr})+E$, then it holds that there are at least $k+1$ elements in $\mathrm{arr}^{\prime}$ that is larger than $\mathrm{mediam}(\mathrm{arr})+E$ due to the definition of median. Because $|\sigma_{ij}|\leq E,$ for all $i,j$, there are at least $k+1$ elements in $\mathrm{arr}$ that is larger than $\mathrm{mediam}(\mathrm{arr})$, which leads to contradiction.
As a result, we have $$\widehat{f}(x)-\widetilde{f}(x)=\mathrm{median}(\mathrm{arr}^{\prime})-\mathrm{mediam}(\mathrm{arr})\leq E.$$
Similarly, $\widehat{f}(x)-\widetilde{f}(x)\geq -E.$
Combining these two results, for all $x$, it holds that
$$|\widehat{f}(x)-\widetilde{f}(x)|\leq E.$$

Finally, we prove the point-wise bound. By the property of CountSketch~\citep{charikar2002finding}, we have for any item $x$, with probability $1-\frac{\beta}{2}$,
$$|\widetilde{f}(x)-f(x)|\leq \gamma\cdot N.$$
Therefore, according to triangle inequality, for each item $x$, with probability $1-\beta$,
$$|\widehat{f}(x)-f(x)|\leq \gamma\cdot N+E.$$
Then the proof is completed by plugging in the definition of $E$.
\end{proof}

\section{Missing Quantile Algorithms}\label{sec:missingQuantile}
\citet{gilbert2002summarize} made the connection between frequency and quantiles to propose the first universe based $RSS$ quantile sketch in the turnstile model.
Observe, that the relationship between frequency and rank is that one can sum up all items' frequency in the range of 0 to the item itself to estimate the rank. However, this naive approach requires summing all items' frequencies in the range, and the error quickly escalates. A better approach is to break the range from 0 to item $x$ into at most $\log U$ dyadic intervals~\citep{cormode2019answering} and then sum all frequencies for each dyadic interval to obtain the estimation of the rank($x$).~\citet{cormode2005improved} proposed the Dyadic Count-Min sketch (DCM) which uses Count-Min sketches for estimating the frequencies of each dyadic interval with space complexity $O(\frac{1}{\gamma} \log^{2}{U} \log{\frac{\log{U}}{\gamma}})$ and update time $O( \log{U}\log{\frac{\log{U}}{\gamma}})$. Later, \citet{wang2013quantiles} leveraged the unbiased property of CountSketch and proposed the Dyadic CountSketch (DCS) which replaces the Count-Min sketch with the CountSketch~\citep{charikar2002finding} to further improve the space complexity to $O(\frac{1}{\gamma} \log^{1.5}{U \log^{1.5}(\frac{\log{U}}{\gamma})})$ while maintaining the same update time. DCS and DCM share the same update and query algorithms as shown below. DCM uses $O(\frac{1}{\gamma} \log{U} \log{\frac{\log{U}}{\gamma}})$ space for each Count-Min sketch and DCS uses $O(\frac{1}{\gamma} \log^{0.5}{U \log^{1.5}(\frac{\log{U}}{\gamma})})$ space for each CountSketch, where both of them use $O(\log\frac{\log U}{\gamma})$ rows.  For more specific details, \citep{cormode2020small} provide a comprehensive analysis of linear and quantile sketches.

\begin{algorithm}[H]
	\caption{DCS/DCM Update$(x, v)$}
	\label{alg:QuantileUpdate}
	\small{
		\begin{algorithmic}[1]
			\STATE {\bfseries Input:} Item $x$ with value $v \in \{-1, +1\}$, and an array of linear sketches $\{\mathrm{LS}_0, \ldots, \mathrm{LS}_{\log U}\}$.
			
			 \FOR{$j \gets 0, \ldots,\log U$}
			        \STATE $\mathrm{LS}_{j}.\mathrm{update}(x, v)$
			         \STATE $x \gets \lfloor x/2 \rfloor$
			\ENDFOR

			\STATE {\bfseries Output: $\{\mathrm{LS}_0, \ldots, \mathrm{LS}_{\log U}\}$. } 
			
		\end{algorithmic}
	}
\end{algorithm}

\begin{algorithm}[H]
	\caption{DCS/DCM Query$(x)$}
	\label{alg:QuantileQuery}
	\small{
		\begin{algorithmic}[1]
			\STATE {\bfseries Input:} Item $x$, and an array of linear sketches $\{\mathrm{LS}_0, \ldots, \mathrm{LS}_{\log U}\}$.
			\STATE $R \gets 0$
			\FOR{$i \gets 0, \ldots,\log U$}
			    \IF{$x$ is odd}
			        \STATE $R \gets R + \mathrm{LS}_{j}.\mathrm{query}(x)$
			     \ENDIF
			     \STATE $x \gets \lfloor x/2 \rfloor$
			\ENDFOR

			\STATE {\bfseries Output: $R$. } 
			
		\end{algorithmic}
	}
\end{algorithm}

Based on the observations, DCS and DCM quantile sketches keeps $\log U$ number of linear sketches, one for each dyadic interval. As a result, to update an item $x$ with value $v \in \{-1, +1\}$, DCS and DCM need to update $\log U$ levels: they first map item $x$ to a dyadic interval for the level and then update the corresponding linear sketch, as shown in Algorithm~\ref{alg:QuantileUpdate}. To estimate the rank of an item, DCS and DCM first break the range into at most $\log U$ dyadic intervals and then query the frequency for each interval from the corresponding linear sketch, as shown in Algorithm~\ref{alg:QuantileQuery}.

\section{Detailed Comparison for Private CountSketch}\label{app:compare}
Recall that for some item $x$, $\widehat{f}(x)$ is the output estimated frequency, $f(x)$ is the actual frequency and $\widetilde{f}(x)$ is the non-private estimated frequency (the output of the original Count-Min and CountSketch with the same set of hash functions). Let $k$ denote the number of rows in our counter, and we choose $k=\log\frac{2}{\beta}$ to bound the failure probability by $\beta$. We first state both \citet{pagh2022improved}'s results and ours.

Our uniform bound (Theorem \ref{thm:median}): $\sup_x |\widehat{f}(x)-\widetilde{f}(x)|\leq E=\widetilde{O}(\sqrt{\frac{k}{\rho}})$.\\
Our lower bound (Theorem \ref{thm:lower}): $\sup_x |\widehat{f}(x)-\widetilde{f}(x)|\geq \Omega(\sqrt{\frac{k}{\rho}})$.\\
\citet{pagh2022improved}'s point-wise bound\footnote{We reformulate the bound in \citep{pagh2022improved} for comparison.}: $|\widehat{f}(x)-f(x)|\leq \text{non-private error bound} + \widetilde{O}(\sqrt{\frac{1}{\rho}})$.

There is a major difference between the analysis in \citep{pagh2022improved} and ours. While \citet{pagh2022improved} focused on the point-wise bound for $|\widehat{f}(x)-f(x)|$ , we focused on the uniform bound for $\sup_{x}|\widehat{f}(x)-\widetilde{f}(x)|$. We are interested in the trade-off between privacy and accuracy for the CountSketch. In particular, we want to answer the question of what additional error will be imposed on the estimated frequency due to the Differential Privacy guarantee. As a result, our result shows an uniform bound $\sup_{x}|\widehat{f}(x)-\widetilde{f}(x)| \leq E$ where $E$ is a function of the desired accuracy, failure probability, and privacy guarantee. To derive the point-wise bound for $|\widehat{f}(x)-f(x)|$, we can simply combine our result with any point-wise bound for $|\widetilde{f}(x)-f(x)|$ (like the one in our work or \citep{pagh2022improved}) due to triangular inequality. However, \citet{pagh2022improved}'s analysis can not imply even point-wise bound for $|\widehat{f}(x)-\widetilde{f}(x)|$.

We agree that \citet{pagh2022improved} has a tight point-wise bound for $|\widehat{f}(x)-f(x)|$ by using the concentration of median. By comparing our analysis and \citep{pagh2022improved} for $|\widehat{f}(x)-f(x)|$, \citet{pagh2022improved}'s point-wise bound removes the $\sqrt{\log\frac{2}{\beta}}$ ($\beta$ denotes failure probability) in the lower order term $E$ which is added to the original estimation error $|\widetilde{f}(x)-f(x)|$. However, the difference may not be substantial. When the database is large (which is the usual case for why we need to perform approximations), $E$ is small compared to the $|\widetilde{f}(x)-f(x)|$. Even for the extreme case of setting $\beta = 1e-10$, the amplification factor for calculating $E$ is $\sqrt{\log\frac{2}{\beta}} \approx 5.8$, which $E$ is still very likely to be small compare to $|\widetilde{f}(x)-f(x)|$.

We believe that none of the two results dominate each other. Both \citet{pagh2022improved}'s point-wise bound for $|\widehat{f}(x)-f(x)|$ and our uniform bound for $|\widehat{f}(x)-\widetilde{f}(x)|$ are useful analysis for understanding the Differentially Private CountSketch with Gaussian noise. 

\section{Extension to the Data Stream Setting}\label{app:extension}
Our Differentially Private Linear Sketches only guarantee Differential Privacy for the query of database, i.e., the adversary is only allowed to query after the whole database passes our algorithm. However, queries for data stream is also quite practical in real-life applications, i.e., the adversary can query at any time. Take reinforcement learning (RL) as an instance, we can use differentially private linear sketches to estimate the visitation number of all (state,action) pairs while preserving privacy. If we only have linear sketches for database, we can only handle offline RL \citep{qiao2022offline}, while with linear sketches for data stream, we can deal with the more challenging online RL \citep{qiao2022sample}.

Our algorithms can be extended to the data stream setting with moderate modifications. Different from our approach of adding noise at the beginning (Algorithm \ref{alg:DP-initialization}), we need to add noise after each item passes our algorithm. To guarantee Differential Privacy under data stream, we can apply the tree-based algorithm (as shown in \citet{kairouz2021practical}) to add Gaussian noises to continuous data. In this way, the algorithm is Differentially Private no matter how many times the adversary queries the data stream and the additional error bound is the same scale as $E$ in our main theorems, with some extra multiplicative logarithmic terms.

\section{Lower Bound for the Additional Error due to Privacy}\label{app:lower}
In this section, we provide a lower bound for CountSketch (the counterpart for Count-Min is similar and we omit it here) showing that our analysis of $\sup_{x}|\widehat{f}(x)-\widetilde{f}(x)|$ is tight. For simplicity, we assume the counter arrays $C$ we use has shape $k\times d$ and $k$ is odd, which means our upper bound $E$ for $\sup_{x}|\widehat{f}(x)-\widetilde{f}(x)|$ is $\widetilde{O}(\sqrt{\frac{k}{\rho}})$. Then according to our Algorithm \ref{alg:DP-initialization}, the noise we add has scale (standard variance) $\sqrt{\frac{k}{\rho}}$. The following theorem shows that if the universe is large enough, with constant probability, $\sup_{x}|\widehat{f}(x)-\widetilde{f}(x)|\geq \Omega(\sqrt{\frac{k}{\rho}})$, which matches our upper bound in Theorem \ref{thm:median} up to logarithmic terms.

\begin{theorem}\label{thm:lower}
There exists constants $c,p$, such that if the size of universe satisfies $U\geq ckd(1+\frac{1}{d-1})^{k-1}$, for our Private CountSketch, there exists a database, with probability at least $p$, $$\sup_{x}|\widehat{f}(x)-\widetilde{f}(x)|\geq \sqrt{\frac{k}{\rho}}.$$
\end{theorem}
\begin{proof}[Proof of Theorem \ref{thm:lower}]
Fix some item $x$ with its corresponding $\{h_i(x),g_i(x)\}_{i\in[k]}$. For each $i\in[k]$, we aim to find $y_i$ in universe such that $h_i(y_i)=h_i(x)$, $h_j(y_i)\neq h_j(x),\,\forall\,j\neq i$ and $g_i(y_i)=g_i(x)$ (this is replaced with $g_i(y_i)=-g_i(x)$ if $i\geq\frac{k+1}{2}$). For any item in the universe, due to the uniform randomness of its hash functions, the probability it satisfies such conditions is $Pr=\frac{1}{2d}(1-\frac{1}{d})^{k-1}$. Therefore, when $U\geq ckd(1+\frac{1}{d-1})^{k-1}$ for some constant $c$, with constant  probability, we can find $\{y_i\}_{i\in[k]}$ satisfying the previous conditions.

Next we can construct a database with only $\{x,y_1,\cdots,y_{k-1}\}$, where the frequency of $x$ is some large $n_x$, and the frequencies of all $y_i$'s are $n_y$ that satisfies that $n_y$ is much larger than $E$. Therefore, the arr of $x$ (without adding noise) consists of one $n_x$, $\frac{k-1}{2}$ numbers much larger than $n_x+E$ and $\frac{k-1}{2}$ numbers much smaller than $n_x-E$. Finally, with high probability, the change from $\widetilde{f}(x)$ to $\widehat{f}(x)$ is from $n_x$ to $n_x+\mathcal{N}(0,\frac{k}{\rho})$, which finishes our proof.
\end{proof}

\end{document}